%% file: paperWCC.tex
\documentclass[a4paper]{llncs}
\usepackage{graphicx}
\usepackage[table]{xcolor}
\usepackage{amssymb}
\usepackage[cmex10]{amsmath}
\usepackage{array}
\usepackage[ruled,vlined,titlenumbered,linesnumbered]{algorithm2e}
\input{macros.tex}

\begin{document}
\title{Cyclic Codes with Locality and Availability\thanks{The work of L. Holzbaur, R. Freij-Hollanti, and A. Wachter-Zeh was supported by the Technical University of Munich -- Institute for Advanced Study, funded by the German Excellence Initiative and European Union 7th Framework Programme under Grant Agreement No. 291763 and the German Research Foundation (Deutsche Forschungsgemeinschaft, DFG) under Grant No. WA3907/1-1. {This work was partly done while R. Freij-Hollanti was with TUM.}}}
\author{Lukas Holzbaur, Ragnar Freij-Hollanti, Antonia Wachter-Zeh}

%
%
\author{Lukas Holzbaur\inst{1} \and 
Ragnar Freij-Hollanti\inst{2} \and
Antonia Wachter-Zeh\inst{1}}
%
%
\institute{Insitute for Communications Engineering, Technical University of Munich, Germany\
  \email{\{lukas.holzbaur, antonia.wachter-zeh\}@tum.de}
  \and
Department of Mathematics and Systems Analysis, Aalto University, Finland\
\email{ragnar.freij-hollanti@aalto.fi}}
\maketitle              
\begin{abstract}
In this work codes with availability are constructed based on the cyclic \emph{locally repairable code} (LRC) construction by Tamo et al. and their extension to $(r,\rho)$-locality by Chen et al. The minimum distance of these codes is increased by carefully extending their defining set. We give a bound on the dimension of LRCs with availability and orthogonal repair sets and show that the given construction is optimal for a range of parameters.

\end{abstract}

\section{Introduction}

Modern databases require coding solutions that offer a trade-off between storage overhead, efficient repair and accessibility. While it is known that maximum distance separable (MDS) code offer the lowest possible storage overhead, they do not perform well with regard to the remaining requirements. In recent years several approaches have been taken to address these issues, one of which being the concept of $r$-locality~\cite{Gopalan2012,Tamo2014,Tamo2016cyc,Cadambe2013,Rawat2016,Tamo2016}, where the number of nodes required to repair a failed node is limited to $r$. In order for multiple simultaneous failures to be repairable, this was extended to \emph{$(r,\rho)$-locality}~\cite{Kamath2014,Tamo2014,Chen2018}, where the code restricted to a subset of positions gives a code of distance at least $\rho$. While these \emph{locally repairable codes} (LRC) address the issue of efficient repair, locality only has a limited effect on accessibility, \emph{i.e.}, situations where the same symbol is requested multiple times simultaneously. In the replicated storage case, each node could be used to serve any request, which is not possible in coded databases. To mitigate this problem \emph{availability} is introduced~\cite{Pamies-Juarez2013,Tamo2014,Tamo2016,Rawat2016,Huang2015}, which requires every symbol in the code to have not only one, but multiple (disjoint) recovery sets. Several bounds on the distance~\cite{Gopalan2012,Kamath2014,Huang2015,Rawat2016,Tamo2016}, rate~\cite{Pamies-Juarez2013,Tamo2016} and dimension~\cite{Cadambe2013} have been derived, both alphabet-dependent and Singleton-like, showing the fundamental limits of such codes. In this work we present a new code construction with locality and availability, based on the cyclic representation of \emph{Tamo-Barg} LRCs~\cite{Tamo2014} given in~\cite{Tamo2016cyc} and their generalization to $(r,\rho)$-locality of~\cite{Chen2018}. We show how the Hartmann-Tzeng bound~\cite{Hartmann1972} can be used to obtain lower bounds on the minimum distance of the local and global codes. Further we derive a bound on the dimension and, for a selection of parameters, compare it to the dimension achieved by our construction.

\section{Preliminaries}
We denote by $[a,b]$ the set of integer $\{i, a\leq i \leq b\}$ and abbreviate it by $[b]$ if $a=0$. A linear code $\code$ of length $n$ and dimension $k$ is denoted by $[n,k]$.

\subsection{Cyclic Codes}

Cyclic codes are a class of linear codes with the property that every cyclic shift of a codeword is also a codeword. Formally an $[n,k]$ cyclic code $\code$ is an ideal in the ring $\F_q[x] / [x^n-1]$ generated by a \emph{generating polynomial} $g(x)$, with $g(x) | x^n-1$. The \emph{defining set} of the code is 
\begin{equation*}
  \mathcal{D}_\code = \{\alpha^i | g(\alpha^i) = 0 \},
\end{equation*}
where $\alpha$ is an $n$-th root of unity from the splitting field $\F_{q^s}$ and $n|(q^s-1)$. Several lower bounds on the distance of cyclic codes have been derived based on this defining set, such as the BCH and the Hartmann-Tzeng bound.
\begin{theorem}[Hartmann-Tzeng bound~\cite{Hartmann1972}] \label{thm:hartmannTzeng}
  Let $\code$ be an $[n,k]$ cyclic code of distance $d$ with defining set $\defset = \{\alpha^{u+s_1z_1+s_2z_2}, 0\leq s_1 \leq \delta-2, 0\leq s_2 \leq \gamma\}$, where $\gcd(n,z_1) = \gcd(n,z_2)=1$. Then $d\geq \delta+\gamma$.
\end{theorem}
The BCH bound is the special case of the Hartmann-Tzeng bound where $\gamma=0$.

\subsection{Locally Repairable Codes}

Locality is a code property that determines the maximal number of symbols required to correct a single erasure. In this work we consider the more general concept of $(r,\rho)$-locality.
\begin{definition}[$(r,\rho)$-locality]\label{def:rrhoLocality}
  Let $\mathcal{A}$ be a partition of $[n-1]$ into sets $A_i$ with $A_i\leq r+\rho-1$. Then a code $\code$ is said to have $(r,\rho)$-locality if for the distance of the code restricted to the positions indexed by $A_i$ it holds that $d(\code|_{A_i}) \geq \rho, \;\forall \, i$.
\end{definition}
As shown in~\cite{Kamath2014}, the distance $d$ of an $[n,k]$ code with $(r,\rho)$-locality is upper bounded by
\begin{equation}\label{eq:SingletonLikeBound}
  d\leq n-k+1 -\left(\left\lceil\frac{k}{r}\right\rceil-1\right)(\rho-1).
\end{equation}
If every symbol has multiple repair sets, this is referred to as availability. Usually it is assumed that all these repair sets are of the same parameters $r$ and $\rho$. We generalize this concept to allow different parameters for the repair sets.
\begin{definition}[$(r_i,\rho_i)$-$t$-Availability] \label{def:rirhoiAvailability}
  Let $\code$ be an $[n,k]$ code. Denote by $\mathcal{A}_i, 1\leq i \leq t$ partitions of $[n-1]$ into sets $A_{i,j}$ with $|A_{i,j}| \leq r_i + \rho_i+1$, such that for every $l \in [n]$ there are indices $j_1,...,j_t$ with
  \begin{equation*}
    \{l\} \in \bigcap\limits_{i=1}^t A_{i,j_i} 
  \end{equation*}
  and
  \begin{equation*}
    d\left(\right.\code\left|_{A_{i,j_i}} \right) \geq \rho_i .
  \end{equation*}
\end{definition}
Note that for $t=1$ this reduces to Definition~\ref{def:rrhoLocality}. In~\cite{Tamo2014}, the partitions are defined to be pairwise orthogonal, \emph{i.e.}, they only intersect in one position. We define a stronger notion of orthogonality and refer to it as \emph{strong orthogonality}.

\begin{definition}[Strong Orthogonality]\label{def:strongOrthogonality}
  Define $t$ partitions of $[n=\prod_{i=1}^t n_i -1]$ given by $\mathcal{A}_i = \{A_{i,1},...,A_{i,\frac{n}{n_i}}\}$, with $|A_{i,j}| = n_i$, as strongly orthogonal if there is an bijection 
  \begin{equation*}
    [n-1] \to \mathbb{Z}/n_1 \times ... \times \mathbb{Z}/n_t
  \end{equation*}
  such that for all $x \in A_{i,j}$
  \begin{equation*}
    x \mapsto (a_1,...,a_i(x),...,a_t) 
  \end{equation*}
  where $a_1,...,a_{i-1},a_{i+1},...,a_t$ are independent of $x$.
\end{definition}
We say that a code with $(r_i,\rho_i)$-$t$-availability and partitioning that fulfills Definition~\ref{def:strongOrthogonality} has \emph{strong $(r_i,\rho_i)$-$t$-availability}. Note that this map corresponds to arranging the symbols of a the repair sets in an LRC on the lines of a $t$-dimensional product code. Take ,\emph{e.g.,} the three partitions $\mathcal{A}_1=\{\{0,1\},\{2,3\},\{4,5\},\{6,7\}\}$, $\mathcal{A}_2 = \{\{0,2\},\{1,3\},\{4,6\},\{5,7\}\}$ and $\mathcal{A}_3 = \{\{0,4\},\{1,5\},\{2,6\},\{3,7\}\}$. Then a valid map would be given by the binary representation of each element, \emph{i.e.}, if $x=x_0+2x_1+4x_2, x_i \in \F_2$ then $x\mapsto (x_0,x_1,x_2)$. 

The construction introduced in the following section is based on Tamo-Barg LRCs~\cite{Tamo2014,Tamo2016cyc}. We focus on cyclic Tamo-Barg codes, and therefore recap the representation given in~\cite{Tamo2016cyc}. The generalization to cyclic codes with $(r,\rho)$-locality of~\cite{Chen2018} is given together with the code construction in Section~\ref{sec:codeConstruction}. In \cite{Tamo2016cyc} it was shown that a $q$-ary LRC is obtained from the defining set
\begin{equation*}
\mathcal{D} = \{\alpha^i, i\mod (r+1) = j\} \cup \{\alpha^{l+sb}, s=0,...,n-\frac{k}{r}(r+1)\} \ ,
\end{equation*}
where $l \mod (r+1) = j$, $\alpha$ is an $n$-th root of unity, and $(n,b)=1$. It was shown that this code has $(r,2)$ locality and meets the Singleton-like bound~\cite{Gopalan2012,Kamath2014} on the distance of $q$-ary LRCs of length $n$ and dimension $k$, \emph{i.e.}, is an optimal $[n,k]$ LRC. In the following we will consider cyclic Tamo-Barg codes with $(r,\rho)$-locality and availability. As shown in~\cite{Tamo2014}, these properties can also be obtained by a construction based on the monomials of the evaluation polynomial. For these codes the minimum distance is given by the minimum distance of the Reed-Solomon containing the LRC.

\section{Code Construction} \label{sec:codeConstruction}

In this section we introduce the code construction considered in this work. It is a combination of the construction of Tamo-Barg codes with availability and $(r,\rho)$-locality given in~\cite{Tamo2014} their cyclic representation as introduced in~\cite{Tamo2016cyc} and generalized in~\cite{Chen2018}.

\begin{definition} \label{def:codeconstruction}
  Let $n_1,...,n_t$ be integers with $n_i\geq 2$ and $\gcd(n_i,n_j) = 1$. Let $\rho_1,...,\rho_t$ be integers with $2 \leq \rho_i \leq n_i$. Let $n= \prod_{i=1}^{t} n_i$ and
  \begin{equation*}
    \defset_i = \left\{ \alpha^{j n_i + sb_i + l} , 0\leq j \leq\frac{n}{n_i}-1, 0\leq s \leq \rho_i-2 \right\},
  \end{equation*}
  where $(n_i,b_i) = 1$. Define the $q$-ary cyclic code $\code$ by its defining set
  \begin{align*}
    \mathcal{D} = \bigcup\limits_{i=1}^{t} \mathcal{D}_{i} \bigcup \defset_g 
  \end{align*}
  where $\defset_g \subseteq [n-1]$.
\end{definition}

Similar to the construction in~\cite{Tamo2016cyc}, the sets $\defset_i$ give the locality, while the set $\defset_g$ is used to increase the global distance of the code. 
The following expands Lemma~3.3 from~\cite{Tamo2016cyc} by showing that the code given in Definition~\ref{def:codeconstruction} has strong $(n_i-\rho_i+1,\rho_i)$-$t$-availability. 
\begin{theorem}\label{thm:cyclicAvailability}
  The cyclic code $\code$ from Definition~\ref{def:codeconstruction} has strong $(n_i-\rho_i+1,\rho_i)$-$t$-availability and distance $d \geq \prod_{i=1}^t \rho_i$.
\end{theorem}

The proof follows from Lemma~\ref{lem:cyclicLocality} and is given at the end of the section. We first show locality properties of the construction given in Definition~\ref{def:codeconstruction} without availability, \emph{i.e.}, for $t=1$. The following result was already shown in~\cite{Chen2018}, for completeness we include a slightly different proof.
\begin{lemma}\label{lem:cyclicLocality}
  The cyclic code $\code$ with defining set $\defset_i$ as in Definition~\ref{def:codeconstruction} has $(n_i-\rho_i+1,\rho_i)$-locality.
\end{lemma}
\begin{proof}
  We show that puncturing the code in all except~$n_i$ positions gives a cyclic code of length $n_i$ and distance $\rho_i$. The defining set of the dual code $\code^\perp$ is given by $\defset_i^\perp = \{\alpha^i, 0 \leq z \leq n-1\} \backslash \defset_i$. Let $\nu=\frac{n}{n_i}$ and $\mathcal{S}= [n]\backslash \{z\nu, 0\leq z \leq n_i-1\}$. Define shortening in a position $z$ as the removal of all codewords with $c_z \neq 0$ from the codebook. We start by showing that the code $\code_S^\perp$, which is $\code^\perp$ shortened in $\mathcal{S}$, is of dimension $\rho_i-1$. Under the map $y \mapsto x^\nu$ each codeword of the shortened code can be written as
  \begin{equation*}
    c_S^\perp(y) = c_0 + c_1 y + c_2 y^2 + \ldots + c_{n_i-1} y^{n_i-1} \mod y^{n_i}-1 .
  \end{equation*}
  and the defining set of the dual code is given by
  \begin{align*}
    \defset_{i}^\perp &= \{\alpha^{z\nu} , 0\leq z \leq n-1\} \backslash \{ \alpha^{(zn_i+sb_i +l)\nu} , 0\leq z \leq\nu-1, 0\leq s \leq \rho_i-2 \} \\
                             &\stackrel{(a)}{=} \{\alpha^{z\nu} , 0\leq z \leq n_i-1\} \backslash  \{ \alpha^{(sb_i+l)\nu} , 0 \leq s \leq \rho_i-2\} \\
                             &=  \{\alpha^{(zb_i+l)\nu} , \rho_i-1 \leq z \leq n_i-1\} ,
  \end{align*}
  where $(a)$ holds because $\alpha^\nu$ is an $n_i$-th root of unity.
  Since the defining set of a code is preserved by shortening, the shortened code~$\code_S^\perp$ contains the polynomial $c_S^\perp(y)$ if and only if~$c_S^\perp (\beta) = 0 \; \forall \, \beta \in \defset_{i,S}^\perp$. Let
  \begin{equation}\label{eq:cishortened}
    c_j = \alpha^{j(\gamma b_i-l)\nu},  n_i-\rho_i +2 \leq \gamma \leq n_i
  \end{equation}
  then  
  \begin{align*}
    c_S^\perp (\alpha^{(zb_i+l)\nu}) &= \sum_{j=0}^{n_i-1} c_j \left(\alpha^{(zb_i+l)\nu}\right)^j \\
                              &= \sum_{j=0}^{n_i-1} \alpha^{j(\gamma b_i-l)\nu} \alpha^{j(zb_i+l)\nu } \\
                              &= \sum_{j=0}^{n_i-1} \left(\left(\alpha^\nu\right)^{(\gamma+z)b_i}\right)^j \\
                              &\stackrel{(a)}{=} \left\{
                                \begin{array}{ll}
                                  n_i \!\!\mod p,\;\; \;& \mathrm{if} \; n_i|(\gamma+z)b_i \\
                                  0,& \mathrm{else}
                                \end{array} \right. ,
  \end{align*}
  where $(a)$ holds because $\order(\alpha^\nu) = n_i$. Since $\gcd(n_i,b_i) = 1$, the condition $n_i|(\gamma+z)b_i$ is equivalent to $n_i|(\gamma+z)$. Observe that
  \begin{equation*}
\{\gamma+z , n_i-\rho_i+2 \leq \gamma \leq n_i, \rho_i-1 \leq z \leq n_i-1\} = [n_i+1, 2n_i -1]
\end{equation*}
 doesn't contain a multiple of $n_i$, so $c_S^\perp \in \code_S^\perp$.\\
  Clearly the $c_j$ of (\ref{eq:cishortened}) give $\rho_i-1$ independent codewords. Since the code $\code_S^\perp$ is of length $n_i$, has a defining set of cardinality~$|\defset_{i,S}^\perp |= r$ and dimension $\geq \rho_i-1$, it is the cyclic code generated by
  \begin{equation*}
    g_S^\perp(y) = \prod_{\beta \in \defset_{i,S}^\perp} (y-\beta)
  \end{equation*}
  in $\F_q[y]/[y^{n_i}-1]$ and its parity-check polynomial is given by
  \begin{equation*}
    h_S^\perp(y) = \prod_{\beta \in \defset_{i,S}} (y-\beta) ,
  \end{equation*}
  where~$\defset_{i,S}$ is $\defset_i$ under the map $y \mapsto x^\nu$.\\
  It is well known that the dual of a punctured code is the dual code shortened in the same positions. It follows that the local code, \emph{i.e.}, $\code$ punctured in the positions $\mathcal{S}$, is equivalent to the code generated by $h_S^\perp(y)$ and the $(n_i-\rho_i+1,\rho_i)$-locality follows by the BCH bound.
  \qed
\end{proof}

Using this result we now give the proof of Theorem~\ref{thm:cyclicAvailability}.
\begin{proof}[Theorem~\ref{thm:cyclicAvailability}]
  By Lemma~\ref{lem:cyclicLocality} and the fact that $\defset_i \subseteq \defset$ it follows that $\code$ has locality $(n_i-\rho_i+1, \rho_i), \; 1\leq i \leq t$. To fulfill the definition of availability, we show that the partitioning of the code is strongly orthogonal as in Definition~\ref{def:strongOrthogonality}. By Lemma~\ref{lem:cyclicLocality} the sets $A_{i,j}$ of partition $\mathcal{A}_i$ are given by
  \begin{equation*}
    A_{i,j} = \left\{j + v \frac{n}{n_i} , 0\leq v \leq n_i-1\right\}
  \end{equation*}
  with $j\in [n_i-1]$. Since $n_l \left| \frac{n}{n_i}\right. \; , \forall \, l \neq i$, mapping $a \in A_{i,j}$ by
  \begin{align*}
    a &\mapsto \{a \!\! \mod n_1,...,a\!\! \mod n_i, ..., a\!\! \mod n_t\}\\
    &= \{j\!\!  \mod n_1,..., a \!\! \mod n_i ,...,j\!\! \mod n_t\}
  \end{align*}
  fulfills the definition of strong orthogonality, as $j$ is constant for all $a \in A_{i,j}$.\\
  Noticing that any code that fulfills strong orthogonality is a $t$-dimensional product code gives the bound on the distance.
  \qed
\end{proof}
Note that the bound on the distance given in Theorem~\ref{thm:cyclicAvailability} is independent of the choice of the set $\defset_g$. The purpose of this set is increasing the global distance of the code as will be discussed later. It follows from the product code like structure given by the strong orthogonality of the partitions, that for $\defset_g = \emptyset$ the dimension of a code~$\code$ as in Definition~\ref{def:codeconstruction} is 
\begin{equation*}
  k = \prod_{i=1}^t (n_i-\rho_i+1 ).
\end{equation*}
For $t=2$ the bound of $d\geq4$ on the distance can also be shown by the Hartmann-Tzeng bound (see Theorem~\ref{thm:hartmannTzeng}).

\begin{lemma}\label{lem:sij}
  Let $s_{i,j} = jq-ip \mod pq$. Then $s_{i,j} = - s_{q-i,p-j}$.
\end{lemma}
\begin{proof}
  Subtracting gives
  \begin{align*}
    s_{i,j} - ( -s_{q-i,p-j}) &= jq - i p + ((p-j) q- (q-i)p) \mod pq\\
    &= jq - i p + (pq - jq - qp + ip) \mod pq = 0
  \end{align*}
  \qed
\end{proof}

\begin{lemma}\label{lem:basiccode}
  For $t=2$ the lower bound of $d\geq 4$ on the distance of a code as in Definition~\ref{def:codeconstruction}, with $\rho_1=\rho_2=2$, is given by the Hartmann-Tzeng bound.
\end{lemma}
\begin{proof}
Assume $l=0$. For the Hartmann-Tzeng bound, as given in Theorem~\ref{thm:hartmannTzeng}, let $u=i n_1, z_1 = s_{i,j_1}$ and $z_2 = s_{i,j_2}$, where $0<i<n_2$ and $0<j_1<n_1$ and $j_2 =n_1- j_1$. By Definition~\ref{def:codeconstruction} it holds that $u\in \mathcal{D}_{1}$ and $u+z_1, u+z_2 \in \mathcal{D}_{2}$. Lemma~\ref{lem:sij} gives
  \begin{align*}
    u+z_1+z_2 &= u+ s_{i,j_1} + s_{i,j_2} \mod n \\
              &= u+ s_{i,j_1} + s_{i,n_1-j_1} \mod n \\
              &= u+ s_{i,j_1} - s_{n_2-i,j_1} \mod n\\
              &= in_1 + j_1n_2-in_1 -j_1n_2+(n_2-i)n_1 \mod n\\
              &= (n_2-i)n_1 \mod n \quad \in \mathcal{D}_{1}.
  \end{align*}
  It holds that $n_1,n_2\nmid s_{i,j}$ for any $i \mod n_2 \neq 0$ and $j \mod n_1 \neq 0$. Since $n_1$ and $n_2$ are the only divisors of $n$ it follows that $\gcd(n,z_1) = \gcd(n,z_2) = 1$. The bound on the distance is then given by  Hartmann-Tzeng bound. From the fact that the Hartmann-Tzeng bound is invariant to a cyclic shift of the defining sets, it follows that the bound is also valid for the case of $l>0$.
\qed
\end{proof}

\section{Bound on the Dimension}

In this section we give a new bound on the dimension of LRCs with strong orthogonality. The approach of the bound is related to the Cadambe-Mazumdar bound~\cite{Cadambe2013} in that it utilizes the dependencies within the local groups to deduce an upper bound on the dimension. Similar to the bound given in \eqref{eq:SingletonLikeBound} our bound is independent of the field size.

\begin{theorem}
  \label{thm:bound2}
  Let $n_1,...,n_t$ be integers with $n_i\geq 2$ and $\gcd(n_i,n_j) = 1$. Let $\rho_1 =...=\rho_t = \rho$ be integers with $2 \leq \rho \leq \min\{n_1,...,n_t\}$. Let $\code$ be a code of length $n$ with strong $(n_i-\rho+1,\rho)$-$t$-availability and partitioning given by~$\mathcal{A}_{i}, 1\leq i \leq t$.
  Let $\xi \in [t]$ be the smallest integer such that
  \begin{equation}\label{eq:pxi}
   n_{\xi+1} > \left\lfloor\sqrt[t-\xi]{\frac{d-1}{\prod_{i=1}^{\xi} n_i}}\right\rfloor.
  \end{equation}
  Then the code dimension is upper bounded by
  \begin{align}
    \label{eq:boundmulticube}
    k &\leq \prod_{i=1}^{t}(n_i-\rho+1) -  \left( \left\lfloor \left(\frac{d-1}{ \prod_{i=1}^{\xi} n_i} \right)^{\frac{1}{t-\xi}}\right\rfloor -\rho-1 \right)^{t-\xi} \cdot \prod_{i=1}^{\xi} (n_i-\rho+1 ).
  \end{align}
\end{theorem}
\begin{proof}
  We regard the codeword as a $t$-dimensional rectangle (hyperrectangle) with side lengths $n_i, i\in [1,t]$. Each line in a dimension is a local group and therefore the corresponding columns of the generator matrix span a vector space of at most dimension $n_i-\rho+1$. Puncturing a code at $\leq d-1$ positions must not decrease its dimension. It follows that any upper bound on the dimension of a code $\code$ punctured in $d-1$ positions is also an upper bound on the dimension of $\code$. In the following, the code is punctured in positions that form a hyperrectangle that is as close as possible to a hypercube, \emph{i.e.}, each of its sides is of length $\min\{v,n_i \}$, where $v$ is the largest integer such that the volume $V$ (number of punctured positions) of the hyperrectangle is $V \leq d-1$. Let $\xi$ be the largest index such that $n_\xi \leq v$. Then $v$ is given by
  \begin{equation*}
    v = \left\lfloor\sqrt[t-\xi]{\frac{d-1}{\prod_{i=1}^{\xi}n_i}} \right\rfloor ,
  \end{equation*}
  since
  \begin{align*}
   V = v^{t-\xi} \cdot \prod_{i=1}^\xi n_i &= \left(\left\lfloor\sqrt[t-\xi]{\frac{d-1}{\prod_{i=1}^{\xi}n_i}} \right\rfloor \right)^{t-\xi}\cdot \prod_{i=1}^\xi n_i\\
   &\leq \left(\sqrt[t-\xi]{\frac{d-1}{\prod_{i=1}^{\xi}n_i}} \right)^{t-\xi}\cdot \prod_{i=1}^\xi n_i\\
   &= d-1 .                                          
  \end{align*}
  Considering the code given by the columns according to one line in a given dimension gives a code of distance $\rho$ and therefore puncturing the code in any $\rho-1$ position of a line does not decrease the dimension. Since any symbol is part of $t$ repair sets, each of these sets has to be punctured in at least $\rho$ positions to decrease the dimension. It follows that puncturing the code in the positions corresponding to a $t$-dimensional hypercube of side length $v'$ decreases the dimension by the volume of a $t$-dimensional hypercube of side length $v'-(\rho-1)$. If $v' > n_i$ for any $i \in [t]$, the punctured positions can no longer form a hypercube. In this case the punctured positions in these $\xi$ dimensions fill the complete lines, while the remaining dimensions form a $(t-\xi)$-dimensional hypercube of side length $v$. Equation~\eqref{eq:boundmulticube} follows from the volume obtained by subtracting $\rho-1$ from each side length, \emph{i.e.}, taking the linear dependencies in each line into consideration.  
  \qed
\end{proof}

For simplicity the bound is only given for $\rho_1=...=\rho_t=\rho$ but the generalization is straight forward by replacing $\rho$ with $\rho_i$ and taking the product instead of powering.
\begin{remark}\label{rem:remainder}
Note that in the bound as stated above, the code is only punctured in the positions according to the largest $t$-dimensional hyperrectangle possible, while it would be possible to further puncture positions according to a $t-1$-dimensional hyperrectangle, and so on. We refrain from including this in the bound as, in our opinion, the more complicated recursive structure required is not justified by the obtained improvement.
\end{remark}

\section{Examples}
In this section we give the parameters of some cyclic Tamo-Barg codes constructed as described in the previous sections and compare the achieved dimension to the bound for a selection of parameters.

\begin{example}
We design an $[n=15,k=6]$ LRC over $\F_{16}$ with strong $(\{2,4\},\{2,2\})$-$2$-availability. Setting $b_1=b_2=1$ and $l=0$ gives
\begin{equation*}
  \defset = \underbrace{\{\alpha^0,\alpha^3,\alpha^6,\alpha^9,\alpha^{12}\}}_{\defset_1} \cup \underbrace{\{\alpha^0,\alpha^5,\alpha^{10}\}}_{\defset_2} \cup \defset_g 
\end{equation*}
as the defining set of the code. 
With $\defset_g=\{\alpha^7,\alpha^8\}$ the BCH bound gives a lower bound of $d\geq 7$ on the distance of the code. Note that for this example it is not possible to achieve a higher bound on the distance with the Hartmann-Tzeng bound. 
The corresponding defining set can be visualized as follows:

\begin{equation*}
  \settowidth\mylen{100} 
  \begin{array}{P||P|P|P|P|P|P|P|P|P|P|P|P|P|P|P}
    \alpha^i &0&1&2&3&4&5&6&7&8&9&10&11&12&13&14\\ \hline
    \defset_1 &0&&&0&&&0&&&0&&&0&&\\
    \defset_2 &0&&&&&0&&&&&0&&&& \\
    \defset_g &&&&&&&&0&0&&&&&& \\ \hline
    \defset &0&&&0&&0&0&0&0&0&0&&0&
  \end{array}
\end{equation*}
The dimension of this code is $k = n-|\defset| = 6$. The bound of Theorem~\ref{thm:cyclicAvailability} for $d=7$ gives $k\leq 7$, however, when also considering the remainder of the floor operation in the bound, as explained in Remark~\ref{rem:remainder}, the bound gives $k\leq 6$, showing that this example is optimal.
\end{example}

\begin{example}
  In Example~5 of~\cite{Tamo2014} an $[n=12, k=4]$ LRC with strong $(\{2,3\},\{2,2\})$-$2$-availability is constructed over $\F_{13}$ and the distance of the code is lower bounded by $d\geq 6$. For the same parameters our construction with $b_1=b_2=1$ and $l=0$ gives
  \begin{equation*}
      \defset = \underbrace{\{\alpha^0,\alpha^3,\alpha^6,\alpha^9\}}_{\defset_1} \cup \underbrace{\{\alpha^0,\alpha^4,\alpha^{8}\}}_{\defset_2} \cup \defset_g .
  \end{equation*}
  With $\defset_g = \{\alpha^5,\alpha^7\}$ the BCH bound gives $d\geq 8$ for the distance of the cyclic code with dimension $k = n-|\defset| = 4$.
\end{example}

Table~\ref{tab:parameters} gives the parameters for codes defined as in Definition~\ref{def:codeconstruction}, where $l_1=l_2=0$ and $\rho_1=\rho_2=2$. The parameters for which the dimension of the obtained code is optimal for the given distance are highlighted. Note that this does not necessarily imply that the code is also optimal in terms of distance for a given dimension, as the bound can be the equal for different values of $d$. 
\begin{table}\centering
  \caption{Parameters of the code construction given in Definition~\ref{def:codeconstruction} with additional zeros given by $\defset_g$ added to the defining set, which are chosen optimal with respect to the Hartmann-Tzeng bound. Parameters for which the dimension of the respective code attains the bound for the given distance are indicated in green.}
\begin{tabular}{CCCCCCCC}
  n & n_1 & n_2 & D_g &  d & k & \mathrm{Theorem~\ref{thm:bound2}} \\ \hline
  \rowcolor{green!30} 15 & 3 & 5 & \{\alpha^4\} & \geq 5 & 7 & 7\\
  \rowcolor{green!30} 15 & 3 & 5 & \{\alpha^4,\alpha^7,\alpha^8,\alpha^{11}\} & \geq 11 & 4 & 4\\ 
  \rowcolor{green!30} 21 & 3 & 7 & \{\alpha^8\} & \geq 5 & 11 & 11\\
  21 & 3 & 7 & \{\alpha^4,\alpha^5\} & \geq 6 & 10 & 11\\
  \rowcolor{green!30} 21 & 3 & 7 & \{\alpha^8,\alpha^{10},\alpha^{11},\alpha^{13}\} & \geq 11 & 8 & 8\\
  \rowcolor{green!30} 51 & 3 & 17 & \{\alpha^{16}\} & \geq 5 & 31 & 31\\  
  51 & 3 & 17 & \{\alpha^{14},\alpha^{16}\} & \geq 6 & 30 & 31\\
  51 & 3 & 17 & \{\alpha^{10},\alpha^{11},\alpha^{13},\alpha^{14},\alpha^{16}\} & \geq 11 & 27 & 28\\
  35 & 5 & 7 & \{\alpha^4\} & \geq 4 & 23 & 24 \\
  35 & 5 & 7 & \{\alpha^4,\alpha^6\} & \geq 5 & 22 & 23 \\
  35 & 5 & 7 & \{\alpha^8,\alpha^9,\alpha^{11},\alpha^{12},\alpha^{13}\} & \geq 10 & 19 & 20 \\  
\end{tabular}
\label{tab:parameters}
\end{table}

\section{Conclusion}

In this work we gave a code construction with $(r,\rho)$-locality and availability based on the construction of cyclic Tamo-Barg LRCs~\cite{Tamo2014,Tamo2016cyc} and their generalization~\cite{Chen2018}. We introduce the notion of strongly orthogonal repair sets and show that the provided a construction has this property. We give a bound on the dimension of any LRC with strong orthogonality and show that for some parameters, the dimension achieved by our construction is optimal or close to optimal. 
\newpage

%
%
%
\bibliographystyle{splncs04}
\bibliography{paperWCC}

\end{document}

%% file: macros.tex
\newcolumntype{C}{>{$}c<{$}} 

\DeclareMathOperator{\order}{order}

\newcommand{\code}{\ensuremath{\mathcal{C}}}

\newcommand{\F}{\ensuremath{\mathbb{F}}}
\newcommand{\defset}{\ensuremath{\mathcal{D}}}

\newlength\mylen
\newcolumntype{P}{>{\hfil$}p{\mylen}<{$\hfil}} 

%% file: paperWCC.bbl
\begin{thebibliography}{10}
\providecommand{\url}[1]{\texttt{#1}}
\providecommand{\urlprefix}{URL }
\providecommand{\doi}[1]{https://doi.org/#1}

\bibitem{Cadambe2013}
Cadambe, V., Mazumdar, A.: {An upper bound on the size of locally recoverable
  codes}. In: 2013 International Symposium on Network Coding (NetCod).
  pp.~1--5. IEEE (jun 2013)

\bibitem{Chen2018}
Chen, B., Xia, S.T., Hao, J., Fu, F.W.: {Constructions of Optimal Cyclic
  $(r,\delta)$ Locally Repairable Codes}. IEEE Trans. Inf. Theory
  \textbf{64}(4),  2499--2511 (apr 2018)

\bibitem{Gopalan2012}
Gopalan, P., Huang, C., Simitci, H., Yekhanin, S.: {On the Locality of Codeword
  Symbols}. IEEE Trans. Inf. Theory  \textbf{58}(11),  6925--6934 (nov 2012)

\bibitem{Hartmann1972}
Hartmann, C., Tzeng, K.: {Generalizations of the BCH bound}. Information and
  Control  \textbf{20}(5),  489--498 (jun 1972)

\bibitem{Huang2015}
Huang, P., Yaakobi, E., Uchikawa, H., Siegel, P.H.: {Linear locally repairable
  codes with availability}. In: 2015 IEEE International Symposium on
  Information Theory (ISIT). pp. 1871--1875. IEEE (jun 2015)

\bibitem{Kamath2014}
Kamath, G.M., Prakash, N., Lalitha, V., Kumar, P.V.: {Codes With Local
  Regeneration and Erasure Correction}. IEEE Trans. Inf. Theory
  \textbf{60}(8),  4637--4660 (aug 2014)

\bibitem{Pamies-Juarez2013}
Pamies-Juarez, L., Hollmann, H.D., Oggier, F.: {Locally repairable codes with
  multiple repair alternatives}. In: 2013 IEEE International Symposium on
  Information Theory. pp. 892--896. IEEE (jul 2013)

\bibitem{Rawat2016}
Rawat, A.S., Papailiopoulos, D.S., Dimakis, A.G., Vishwanath, S.: {Locality and
  Availability in Distributed Storage}. IEEE Trans. Inf. Theory
  \textbf{62}(8),  4481--4493 (aug 2016)

\bibitem{Tamo2014}
Tamo, I., Barg, A.: {A Family of Optimal Locally Recoverable Codes}. IEEE
  Trans. Inf. Theory  \textbf{60}(8),  4661--4676 (2014)

\bibitem{Tamo2016}
Tamo, I., Barg, A., Frolov, A.: {Bounds on the Parameters of Locally
  Recoverable Codes}. IEEE Trans. Inf. Theory  \textbf{62}(6),  3070--3083
  (2016)

\bibitem{Tamo2016cyc}
Tamo, I., Barg, A., Goparaju, S., Calderbank, R.: {Cyclic LRC codes, binary LRC
  codes, and upper bounds on the distance of cyclic codes}. International
  Journal of Information and Coding Theory  \textbf{3}(4), ~345 (2016)

\end{thebibliography}
